\documentclass[12pt,a4wide]{amsart}

\usepackage{graphicx}
\usepackage{latexsym}
\usepackage{epic}
\usepackage[leqno]{amsmath}
\usepackage{amssymb}

\newtheorem{lemma}{Lemma}[section]
\newtheorem{theorem}[lemma]{Theorem}
\newtheorem{proposition}[lemma]{Proposition}
\newtheorem{corollary}[lemma]{Corollary}
\newtheorem{remark}[lemma]{Remark}
\newtheorem{definition}[lemma]{Definition}
\newtheorem{conjecture}[lemma]{Conjecture}

\newcommand{\R}{{\mathbb R}}

\newcommand{\T}{\mathcal{T}}

\begin{document}

\title[Encoding and Constructing 1-Nested Networks]
{Encoding and Constructing 1-Nested Phylogenetic Networks with Trinets
}

\author{K.~T.~Huber        \and
        V.~Moulton.
}

\thanks{K.~T.~Huber and  V.~Moulton,
   School of Computing Sciences,
University of East Anglia,
Norwich, NR4 7TJ,
United Kingdom.\,\,
              \email{Katharina.Huber@cmp.uea.ac.uk}           
\email{vincent.moulton@cmp.uea.ac.uk}
}

\date{\today}

\maketitle

\begin{abstract}
Phylogenetic networks are a generalization of phylogenetic
trees that are used in biology 
to represent reticulate or non-treelike 
evolution.
Recently, several algorithms have been 
developed which aim to construct 
phylogenetic networks from biological data
using {\em triplets}, i.e. binary phylogenetic trees
on 3-element subsets of a given 
set of species. However, a fundamental problem
with this approach is that the triplets displayed by
a phylogenetic network do not necessary 
uniquely determine or {\em encode} the network.
Here we propose an alternative approach to 
encoding and constructing phylogenetic networks, which 
uses phylogenetic networks on 3-element subsets
of a set, or {\em trinets}, rather than triplets. 
More specifically, we show that for 
a special, well-studied type of phylogenetic network called 
a 1-nested network, the trinets displayed by 
a 1-nested network always encode 
the network. We also present an efficient algorithm 
for deciding whether a {\em dense} set of trinets (i.e. 
one that contains a trinet on every 3-element subset 
of a set) can be 
displayed by a 1-nested network or not and, if
so, constructs that network.
In addition, we discuss some potential new directions
that this new approach opens up for constructing
and comparing phylogenetic networks. 
\end{abstract}



\section{Introduction}

Phylogenetic networks are a generalization of
phylogenetic trees that are used 
in biology to represent reticulate or non-treelike evolution
(cf.  \cite{HRS11,N11} for recent overviews).
There are various types of phylogenetic
networks, but in this paper we shall focus on 
phylogenetic networks that {\em explicitly} 
represent the evolution of a given set of species. Such 
networks (whose formal definition is presented in 
Section~\ref{preliminaries}) can be essentially
regarded as directed acyclic graphs
having a single root, whose 
internal vertices represent ancestral species
and whose leaves represent the set species
(see e.g. Fig.~\ref{dontencode}). 
They have been used, 
just to name a few examples, to represent the evolution
of viruses \cite{SH05}, bacteria \cite{PKFDF03}, 
plants \cite{LSHPOM09}, and fish \cite{KDSASBES07}.

Recently, several algorithms have been 
developed which aim to 
construct phylogenetic networks (cf. \cite{HRS11,N11}). 
However, as stated in \cite[p.xi]{HRS11},
``While there is a great
need for practical and reliable computational 
methods for inferring rooted phylogenetic networks 
to {\em explicitly} describe evolutionary scenarios 
involving reticulate events, generally speaking, such
methods do not yet exist, or have not yet matured 
enough to become standard tools".

Probably one of the main reasons for this 
is that we do not yet have a very good 
understanding of how to build up complex
phylogenetic networks from
simpler structures. An important case in point
is the construction of phylogenetic networks from 
phylogenetic trees. Even though there has been a great deal 
of recent work on this problem (cf. \cite[Chapter 11]{HRS11},
\cite[Section 2]{N11}),
especially concerning the construction of 
networks from {\em triplets} (i.e. binary phylogenetic trees 
with three leaves) \cite{GH11,HIKS11,IKKSHB09,IK11a,JNS06,JS06,TH09}),
there is a fundamental obstacle to this approach: The trees 
displayed by a phylogenetic network do not necessarily determine
or {\em encode} the network \cite{GH11} 
(even on 3 species -- see e.g. Fig.~\ref{dontencode})
and, in fact,  we do not even know 
when a phylogenetic network is uniquely determined by 
{\em all} of the trees that it displays \cite{W11}.

\begin{figure}[b] \centering
\includegraphics[scale=0.5]{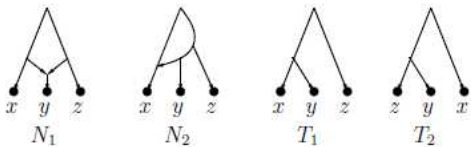}
\caption{Two distinct phylogenetic networks $N_1$ and 
$N_2$ with leaf set $\{x,y,z\}$ that display 
the same set $\{T_1,T_2\}$ of phylogenetic trees.
In particular, neither of these two networks is encoded by this 
set of trees.}
\label{dontencode}
\end{figure}

As an alternative approach to tackling the
problem of constructing phylogenetic networks,
in this paper we shall investigate
the following strategy: Instead of 
constructing phylogenetic networks from trees, try to build them
up from (simpler) phylogenetic networks. More
specifically, we investigate how to 
construct phylogenetic networks from {\em trinets}, 
that is, phylogenetic networks having just three leaves
(see, for example, the networks $N_1$ and $N_2$ in Fig.~\ref{dontencode}).

One of the main difficulties
that we had to overcome before being
able to put this strategy 
into practice was to find an appropriate 
definition for the set of 
trinets that is displayed by a phylogenetic 
network (see Definition~\ref{tripdisp}).
However, with this definition in hand, we 
are able to show that any {\em 1-nested network} -- a 
quite simple and well-studied 
type of phylogenetic network \cite{CLR11} --
is always encoded by the set of
trinets it displays (Theorem~\ref{encode}).
Moreover, using this fact, we 
provide a polynomial-time 
algorithm for deciding whether 
a given {\em dense} set of trinets (i.e. 
one that contains a trinet on every 3-element subset 
of a set) can be 
displayed by a 1-nested network or not and, if
so, constructs that network (see 
Fig.~\ref{algorithm:multicons} and Theorem~\ref{build}).

We now describe the contents of the rest of the paper. 
In Section~\ref{preliminaries}
we introduce some relevant, basic terminology concerning 
phylogenetic networks. In Section~\ref{recoverable-nets}
we define the rather natural concept of a recoverable network, and 
show that, although a phylogenetic network need not be
recoverable in general, a 1-nested network always is. 
In the following section, we show 
that a recoverable phylogenetic network
is 1-nested if and only if all of its displayed
trinets are 1-nested (Theorem~\ref{green}).
Using this fact and certain operations on 1-nested networks that are 
closely related to those presented 
in \cite{CLR11} and 
that are presented in Section~\ref{cherries-catuses-etc},
we then establish Theorem~\ref{encode} in Section~\ref{encoding-1-nested}. 
As a corollary, we obtain a new 
(and efficiently computable) proper metric on the set of
1-nested networks all having the same
leaf set (see Corollary~\ref{metric}). 
In Section~\ref{1-nested-from-dense} we present our main 
algorithm for checking whether or not a dense
set of trinets is displayed by a 1-nested network.
We conclude in Section~\ref{discussion}
with a discussion on some possible future 
directions, including some ideas about
how trinets might be used in practical applications.

\section{Preliminaries}\label{preliminaries}

For the rest of this paper, 
$X$ is a non-empty, finite set (which 
will usually correspond to a set of species
or organisms). For consistency, 
we follow the notation presented in \cite{CLR11}
where appropriate.

An {\em rDAG} $N = (V,A)$ is a 
directed acyclic graph (DAG)
with non-empty vertex set $V=V(N)$, non-empty arc set $A=A(N)$ 
(with no multiple 
arcs) and single root $\rho=\rho_N$ (i.e. a DAG with
precisely one source $\rho$). We 
let $<_N$ denote the usual partial order on $V$ induced by $N$.
The {\em underlying graph} of $N$ is denoted $\underline{N}$.
A {\em cycle} in $\underline{N}$ is a subset 
$C = \{v_1,v_2,\dots,v_n\} \subseteq V(\underline{N})$, $n \ge 3$, 
such that $\{v_i,v_{i+1}\} \in E(\underline{N})$ for
all $1 \le i \le n-1$ and $\{v_1,v_n\} \in E(\underline{N})$.
If $C$ is some cycle in $\underline{N}$
and there is some $v\neq w \in V$ so that 
the union of all of the arcs in $N$ having
both vertices in $C$
is the union of two directed paths in $N$
that both start at $v$ and end at $w$, 
then $v$ ($w$) is called the {\em split (end)} vertex of $C$.
We denote an arc $a\in A$ 
with tail $x$ ($= tail(a)$), and head $y$ 
by $(x,y)$. We call $(x,y)$ a {\em cut arc (of $N$)} if 
the removal of the edge $\{x,y\}$ 
from $E(\underline{N})$ disconnects $\underline{N}$.
A vertex $v \in V$ is a called a {\em leaf} of $N$ if
$indegree(v)=1$ and $outdegree(v)=0$. We denote
the set of leaves of $N$ by $L(N)$. 
Every vertex of $N$ that is neither the root $\rho_N$ 
nor has outdegree 0 is called an {\em interior} 
vertex of $N$. A {\em tree vertex} $v \in V$ 
is an interior vertex of $N$ with $indegree(v)=1$, 
and a {\em hybrid vertex}
$v \in V$ is an interior vertex with $indegree(v) \ge 2$. 
Note that neither the root $\rho_N$ nor a leaf of $N$ 
is a tree vertex and that a hybrid vertex of $N$
cannot be a leaf. 

Now, an {\em $X$-rDAG} is an rDAG $N=(V,A)$ with leaves 
uniquely labeled by the elements in $X$ (i.e. there is a 
map $\phi_N:X \to V$ such that $\phi$ maps $X$
bijectively onto $L(N)$). We will usually 
just assume $L(N)=X$ in case the labeling map 
is clear from the context.
A {\em phylogenetic network $N=(V,A)$ (on $X$)} 
is an $X$-rDAG such that every tree vertex 
has outdegree at least 2 and 
every hybrid vertex has outdegree at least 1.
If $N$ is such a network and 
$N'=(V',A')$ is a phylogenetic network on a 
non-empty finite set $Y$,
then $N$ is {\em isomorphic} to $N'$ if
there is a bijection $\xi:X \to Y$
and a directed graph isomorphism 
$\iota:V \to V'$ between $N$ and $N'$
such that $\phi_{N'} = \iota \circ \phi_N \circ \xi^{-1}$. 
In particular, in case $Y=X$ we consider
$X$ as being a subset of both $V$ and $V'$, 
and hence $N$ is isomorphic to $N'$
if and only if $\iota$ 
restricted to $X$ is the identity map on $X$.

A phylogenetic network $N=(V,A)$ on $X$ is  
\begin{itemize}
\item a {\em bush (on $X$)} if it is 
isomorphic to the phylogenetic network with vertex set 
$V = X \cup \{v\}$, $v \not\in  X$, 
and arc set $A = \{(v,x) \,:\, x \in X\}$,
\item a {\em two-leafed network (on $X$)}
if $X=\{x,y\}$, 
and $N$ is isomorphic to the phylogenetic network 
on $X$ with vertex set 
$V =\{u,v,w,x,y\}$ and arc set
$A = \{(u,w),(u,v),(v,w),(v,x),(w,y)\}$, 
\item {\em binary} if all of its hybrid vertices 
have indegree 2 and outdegree 1
and all of its tree vertices have outdegree 2, 
\item {\em 1-nested} if every pair of
cycles in $\underline{N}$ intersect in 
at most 1 vertex\footnote{Note that in \cite{CLR11},
1-nested networks are defined in such a way 
that every hybrid vertex has indegree 2 -- we
do not make this assumption, but we will use the 
same name rather than introducing another term.},
\item a {\em galled tree} if every 
pair of cycles in $\underline{N}$ is disjoint, 
\item a {\em (rooted) phylogenetic tree} if $\underline{N}$ is a tree, and
\item a {\em trinet} if $|L(N)|=|X|=3$.
\end{itemize}
Note that a 1-nested network $N$ on $X$ with $|X|=1$ is a bush
with arc set 
consisting of precisely one arc, and if $|X|=2$ then $N$ is isomorphic to 
either a two-leafed network or a phylogenetic tree with 2 leaves.

In Fig.~\ref{trinets} we picture the set of
all possible non-isomorphic
1-nested trinets on $\{x,y,z\}$.
If $N$ is a 1-nested trinet on $X$, $|X|=3$,
that is not isomorphic to a phylogenetic tree on $X$,
then we say that $t \in X$ is {\em at the bottom of $N$} 
if it corresponds to one of the vertices 
represented by larger dots in Fig.~\ref{trinets}, 
and we say that $t$ {\em hangs off the side of $N$} if it corresponds 
to one of the vertices represented by a square in that figure
(note that, in particular, there may 
be more than one element at the
bottom of a trinet). 

\begin{figure}[t] \centering
\includegraphics[scale=0.5]{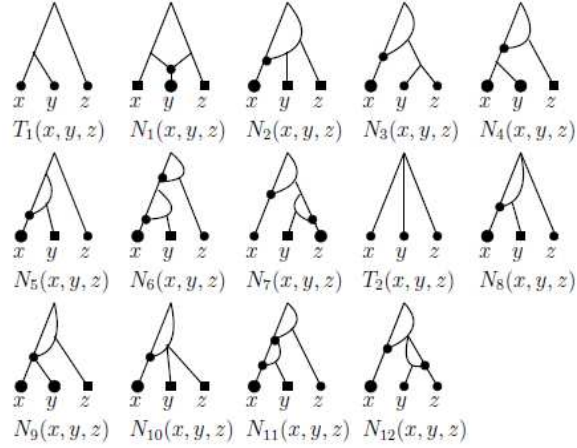}
\caption{The fourteen possible
non-isomorphic, 1-nested trinets on 
the set $\{x,y,z\}$. Directions
on arcs are omitted for clarity; internal vertices 
indicated with a dot are all hybrid vertices. Leaves 
that are at the bottom of a trinet are indicated with large
dots and vertices hanging off the side of a trinet 
with a square.}
\label{trinets}
\end{figure}

Finally, let $\T$ denote a non-empty set of trinets
such that $L(T) \in {X \choose 3}$ for all $T \in \T$
(which we shall also call a {\em trinet set (on $X$)}
for short). 
If $Y \subseteq X$, $|Y|\ge 3$, we let $\T_Y$ be the subset of $\T$ 
consisting of those trinets $T \in \T$ 
with $L(T) \subseteq Y$.
In addition, we call $\T$ {\em dense (on $X$)} 
if ${X \choose 3}= \{ L(N) \,:\, N \in \T \}$ and 
$|\T| = {|X| \choose 3}$.

\section{Trinets and recoverable networks}\label{recoverable-nets}

In this section, we investigate networks that display only
1-nested trinets. In particular, we show that
even if every trinet displayed by a network $N$ is 1-nested, 
it does not necessarily follow that $N$ is 1-nested. 
In addition, we shall introduce a rather natural 
condition on $N$ (that it is a `recoverable network') 
for which this statement does in fact hold (see 
Theorem~\ref{green} in the next section).

Suppose $N=(V,A)$ is a phylogenetic network on $X$, $|X|\geq 3$,
and $Y$ is a non-empty subset of $V-\{\rho_N\}$. 
Let  $v(Y)$ to be the last vertex in $V-Y$ that lies on all paths 
in $N$ from $\rho_N$ to every $y \in Y$.
Note that if $Y$ consists of a single vertex $y$, then $v(\{y\})$ is
known as the {\em immediate dominator of $y$} \cite{LT79}
(see also \cite[p. 143]{HRS11} where it is called the 
{\em lowest stable ancestor of $y$})). 

We now present a key definition (see also Fig.~\ref{big}):

\begin{definition}\label{tripdisp}
Given a phylogenetic network $N$ on $X$ and 
some $Y \in {X \choose 3}$, we define the 
{\em trinet on $Y$ displayed by $N$} to be the trinet $N_Y$
with leaf set $Y$ which is obtained from $N$ by 
first taking the network $\tilde{N}$ consisting of the 
union of all directed paths in $N$ starting
at $v(Y)$ and ending at some element in $Y$, 
and then repeatedly first 
(i) suppressing all vertices $v$ 
with $indegree(v)=outdegree(v)=1$,
and then
(ii) suppressing all multiple arcs that might result, until a 
trinet on $Y$ is obtained.
Put $Tr(N) = \{ N_Y \,:\, Y \in {X \choose 3}\}$.
\end{definition}

Given a phylogenetic network $N$ on $X$, we say that a
trinet set $\T$ on $X$ 
is {\em displayed} by $N$ if $\T\subseteq Tr(N)$.
Moreover, we say that $\T$  {\em encodes} 
$N$ if $\T\subseteq Tr(N)$ and, if $N'$ is any other
 phylogenetic  network on $X$ with  $\T\subseteq Tr(N')$, then 
$N'$ is isomorphic to $N$. 

Note that in Definition~\ref{tripdisp} 
it is necessary to consider (at least) 3-element subsets of $X$, 
since if `binets' are defined in a similar 
way for 2-element subsets, then the resulting 
set would not in general encode the
network (even if the network is a tree). Also note that
we do not define a trinet on $Y$
displayed by $N$ to  be 
the network consisting of the union of all directed paths 
in $N$ to the elements of $Y$ as this can result in 
networks with vertices having  in- and outdegree 1, that is,
networks that are not phylogenetic networks.

\begin{figure}[t] \centering
\includegraphics[scale=0.5]{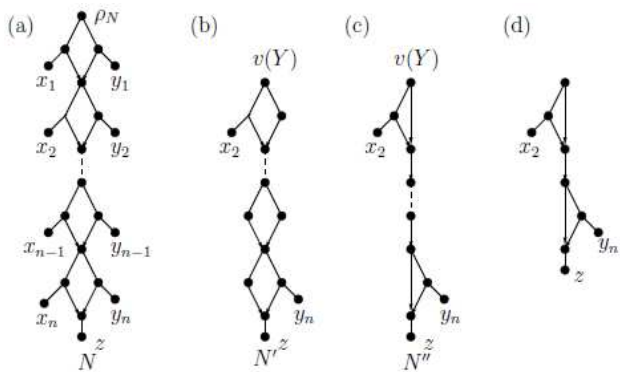}
\caption{(a) A phylogenetic 
network $N$ on $X=\{x_1,\dots,x_n,z,y_1,\dots, y_n\}$, 
$n\geq 1$. (b) The subnetwork
$N'$ obtained by taking the union of the directed
paths from $v(Y)$ to every element in $Y=\{x_2,y_n,z\}$. (c) 
The subnetwork $N''$ obtained from $N'$ by suppressing all multiple 
arcs of $N'$. (d) The trinet obtained from $N''$ by suppressing all vertices 
$v\in V(N'')$ with $indegree(v)=outdegree(v)=1$. Directions of arcs
are omitted when clear.}
\label{big}
\end{figure}

The proof of the following lemma is straight-forward and is omitted:

\begin{lemma}\label{1-nested-implies trinets}
Suppose that $N$ is a 1-nested network on $X$, $|X|\ge 3$.
Then any element in $ Tr(N)$ 
is isomorphic to one of the fourteen 
trinets on $\{x,y,z\}$ presented in Fig.~\ref{trinets}.
\end{lemma}

\begin{remark}
If $N$ is a 1-nested network on $X$, $|X|\ge 3$, then 
$N$ is binary if every element in $Tr(N)$
is isomorphic to either $T_1(x,y,z)$ or one of 
$N_i(x,y,z)$, $1 \le i \le 7$.
Moreover,  binary level-1 networks and galled trees
(as defined in \cite{CLR11})
can be characterized in a similar manner. 
\end{remark}

Now, suppose that $N$ is a phylogenetic network 
on $X$ such that every trinet in $Tr(N)$ is isomorphic 
to one of the fourteen trinets presented 
in Fig.~\ref{trinets}. It is tempting
to think that this should imply that $N$ is
1-nested. However, this is not the case. For example, even 
if $N$ is a phylogenetic network such that 
every trinet in $Tr(N)$ is isomorphic to either 
$T_1(x,y,z)$ or $T_2(x,y,z)$ in Fig.~\ref{trinets}, then $N$ is 
not necessarily isomorphic to a phylogenetic
tree (see  e.g.~Fig.~\ref{not-recoverable}).
\begin{figure}[t] 
\centering
\includegraphics[scale=0.5]{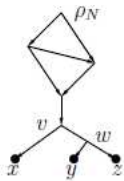}
\caption{A phylogenetic network $N$ on $\{x,y,z\}$ for which 
$Tr(N)$ consists of precisely the trinet $T_1(x,y,z)$ but $N$ is not
a phylogenetic tree on $\{x,y,z\}$. As before, directions 
are omitted for clarity when clear. Also only the vertices that
are leaves are marked by a dot.
 }
\label{not-recoverable}
\end{figure}
Even so, as we shall show in the 
next section (see Theorem~\ref{green}), 
the aforementioned statement is almost correct.

To this end, we now introduce a special class of networks.
Suppose that $N$ is a phylogenetic network
on $X$ with $|X|\geq3$.  We say
that a vertex $v\in V(N)$ is {\em reachable} from a vertex $w\in V(N)-v$
if there exists a directed path in $N$ starting at $w$ and ending in $v$.
In addition, if $z\in V(N)$ is a vertex of $N$ that lies on that path then 
we say that $v$ is reachable from $w$ by {\em crossing} $z$.
We denote by $v_N^*\in V(N)$ 
the (necessarily unique) vertex in $N$ for which there exist 
some distinct $x,y\in X$ with $v_N^*=v(\{x,y\})$
and, for all $\{u,v\}\in{V(N)\choose 2}-\{x,y\}$, either  
$v^*_N=v(\{u,v\})$ holds or $v(\{u,v\})$ is reachable from $v^*_N$.

Now, we say that $N$ is {\em recoverable} if $\rho_N=v_N^*$. 
We use the term recoverable, since for biological data 
it would not be possible to infer the structure of the network above 
$v_N^*$ in case $N$ is not recoverable, as 
there would be no way to `detect' vertices 
above $v_N^*$ using any pair of elements in $X$. 
As an illustration, the vertex
$v$ in the phylogenetic network $N$ on $\{x,y,z\}$ pictured in
Fig.~\ref{not-recoverable} is the vertex
$v_N^*=v(\{x,z\})$. Since $v_N^*\not=\rho_N$, $N$
is {\em not} recoverable.

We now characterize recoverable networks $N$ on $X$, $|X|\geq 3$, in terms
of a special type of vertex. A vertex $v\in V(N)$ is a {\em cut vertex} of $N$
if the deletion of $v$ (plus its incident edges) from $\underline{N}$
disconnects $\underline{N}$. We denote the resulting 
graph by $\underline{N}\backslash v$.
If, in addition, there exists a connected component $K$ of 
$\underline{N}\backslash v$ such that $V(K)\cap L(N)=\emptyset$ 
then we call $v$
a {\em separating vertex} of $N$. For example,
in Fig.~\ref{not-recoverable}
$v$ is a separating vertex of $N$ whereas
vertex $w$ is a cut vertex of $N$.

\begin{proposition}\label{green-cutvertex1}
Suppose $N$ is a phylogenetic network on $X$, $|X|\geq 3$. Then
the following statements hold:
\begin{itemize}
\item[(i)] If $N$ is not recoverable then $v^*_N$ is a cut vertex of $N$.
\item[(ii)] $N$ is recoverable if and only if $v^*_N$ is not a 
separating vertex of $N$.
\end{itemize}
\end{proposition}
\begin{proof}
(i) Note first that $\rho_N\not=v^*_N$ as $N$ is not recoverable.
Let $x,y\in X$ distinct such that $v_N^*=v(\{x,y\})$,
and assume for contradiction that  $v^*_N$ is not
a cut vertex of $N$. Then there
must exist some leaf $l\in L(N)$ of $N$ that is reachable from
$\rho_N$ without crossing $v^*_N$. 
Hence, there exists some $z\in\{x,y\}$ such that $v_N^*\not=v(\{l,z\})$
and $v_N^*$ is reachable from $v(\{l,z\})$; a contradiction. Thus, 
$v_N^*$ is a cut vertex of $N$.

(ii) We prove the contrapositive of the statement i.\,e.\,we show that
$N$ is not recoverable if and only if $v^*_N$ is a separating vertex of $N$.
Suppose $\{x,y\}\in{X\choose 2}$ such that
$v^*_N=v(\{x,y\})$. Assume first 
that $N$ is not recoverable. Then
$\rho_N\not=v^*_N$ and, by (i),
$v^*_N$ is a cut vertex of $N$. Hence, for every leaf $l\in L(N)$ of $N$,
every directed path from $\rho_N$ to $l$ must cross $v^*_N$. 
Let $K_{\rho_N}$ denote
the connected component of $\underline{N}\backslash v^*_N$ 
that contains $\rho_N$ in its vertex set. 
Then $V(K_{\rho_N})\cap L(N)=\emptyset$. Thus $v^*_N$ 
is a separating vertex of $N$.
 
Conversely, suppose that $v^*_N$ is a separating vertex of $N$.
Then $v^*_N$ is a cut vertex of $N$ and so
every directed path from $\rho_N$ to a leaf $z\in L(N)$ of $N$ 
must cross $v^*_N$. If $N$ were recoverable then
$\rho_N =v^*_N$ would follow, implying that every vertex in $N$
must lie on a directed path from $v^*_N$ to a leaf of $N$. But then
$V(K)\cap L(N)\not=\emptyset$ for 
every connected component $K$ 
in $\underline{N}\backslash v^*_N$;
a contradiction.
Thus, $N$ cannot be recoverable. 
\end{proof}

It immediately follows that 1-nested networks are always
recoverable:

\begin{corollary}\label{green-cutvertex}
Suppose $N$ is a phylogenetic network on $X$, $|X|\geq 3$. 
If $N$ is 1-nested, then $N$ is recoverable.
\end{corollary}
\begin{proof}
Suppose for contradiction that there exists
a 1-nested network $N$ on $X$ that is not recoverable, that is, 
$\rho_N\not=v^*_N$.
Then, by Proposition~\ref{green-cutvertex1}(i), $v^*_N$ is
a cut vertex of $N$. Hence,
for every leaf $l\in L(N)$ of $N$, every directed path from
$\rho_N\not = v^*_N$ to $l$ must cross $v^*_N$. Since
$outdegree(\rho_N)\geq 2$ and $N$ cannot have multiple
arcs, it follows that there exist  (at least) 
3 distinct directed paths in $N$  from $\rho_N$
to $v^*_N$. But then there must exist two 
cycles in $\underline N$ which
intersect in at least 2 vertices; a contradiction.
\end{proof}

\section{1-nested trinets imply 1-nested networks} \label{1-nested-implies}

In the last section, we proved that if 
$N$ is a 1-nested phylogenetic network on $X$, $|X|\geq 3$,
then $N$ is recoverable. We shall now prove that
if all of the trinets displayed by a recoverable network 
are 1-nested, then the network is 1-nested (Theorem~\ref{green}).

To this end, 
suppose that $N$ is a phylogenetic network on $X$, $|X|\geq 3$,
and that $C$ is a cycle of $\underline{N}$. Put
$$
Z(C)=\{v\in C\,:\, \mbox{there exist }
\{a, a'\}\in{ A(C)\choose 2} \mbox{ with } tail(a)=tail(a')=v 
\}.
$$
Clearly, $Z(C)\not=\emptyset$.

Now, suppose $l\in L(N)$ is a leaf of $N$ that is reachable
from a hybrid vertex of $N$. We denote by $p(l)$ the
number of distinct 
directed paths in $N$ from $\rho_N$ to $l$. Clearly $p(l)\geq 2$.
Moreover, we denote
by $w(l)$ the unique vertex of $N$ distinct from $l$ 
that simultaneously lies on every directed path from $\rho_N$ to $l$ such that 
(i) $w(l)$ is a hybrid vertex of $N$, and
(ii) there is a unique directed path from $w(l)$ to $l$ 
such that every interior
vertex of $N$ on this path is a tree vertex of $N$. To illustrate
these definitions, consider the network $N$ on $\{x,y,z\}$ depicted
in Fig.~\ref{not-recoverable}. Then $w(y)$ is the
unique hybridization vertex of $N$ and $p(y)=3$. 

We now prove some useful, but somewhat technical, results
concerning the set $Z(C)$. 

\begin{proposition}\label{cycles}
Suppose $N$ is a recoverable phylogenetic network on $X$, $|X|\geq 3$,
such that every trinet in $Tr(N)$ is isomorphic to one of the
fourteen trinets on $\{x,y,z\}$ depicted in Fig.~\ref{trinets}.
Then $|Z(C)|= 1$, for all cycles $C$ in $\underline{N}$.
\end{proposition}
\begin{proof} 
Suppose for contradiction that $\underline{N}$ 
contains a cycle $C$ with
$m:=|Z(C)|\geq 2$. Put $Z(C)=\{z_1,\dots, z_m\}$. Since
$C$ is a cycle in $\underline{N}$ there must exist distinct vertices
$h_i\in C$, $1\leq i\leq m$, such that, for all $1\leq i\leq m$,
two of the incoming arcs of $h_i$ are contained in $A(C)$ and $h_i$ can be
reached from $z_i$ and from $z_{i+1}$, $1 \le i \le m$, where
we define $z_{m+1}:=z_1$. 
Moreover for each such
vertex $h_i$ there must exist a leaf $l_i\in L(N)$ of $N$ that 
is reachable from $h_i$. Note that some of the leaves $l_i$ might be the
same (see Fig.~\ref{reticulation-set} for a representation of the 
generic situation in which all leaves $l_i$, $1\leq i\leq m$, are
distinct).
\begin{figure}[t]
\centering
\includegraphics[scale=0.5]{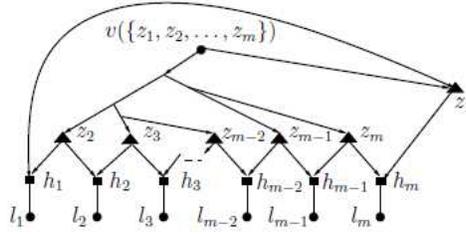}
\caption{The situation considered in the proof of Proposition~\ref{cycles}. 
The vertices in $C$ which have two of their 
incoming (outgoing) arcs contained
in $A(C)$ are marked with squares (triangles). The leaves of $N$
plus the vertex $v(\{z_1,\dots,z_m\}$ are marked by dots. For clarity
all other vertices are not marked.
The directed lines represent directed paths rather than arcs.}
\label{reticulation-set}
\end{figure}

Choose some $i\in\{1,\dots,m\}$, say $i=1$, and let 
$\sigma$ be the ordering $l_1, l_2, \dots,l_m$
of the leaves $l_j$, $1\leq j\leq m$ induced by $C$ via
the vertices $h_i$, $1\leq i\leq m$. If there exist at least
three distinct leaves in that ordering, then let 
$l_{i_1},l_{i_2},l_{i_3}$ denote the first three distinct 
leaves in $\sigma$. Note that $l_1=l_{i_1}$ and each of $l_{i_1}$
 $l_{i_2}$, and $l_{i_3}$ is reachable from 
$v(\{z_1,z_2,\dots, z_p\})\in V(N)$, 
where $p\in \{1,\dots,m\}$  is such that for all $i_3\leq q\leq p$ we have 
$l_q=l_{i_3}$. But then $p\geq 3$ and the 
trinet $N'$ on 
$\{l_{i_1}, l_{i_2},l_{i_3}\}$ displayed by $N$
contains $v(\{z_2,\dots, z_p\})$
in its vertex set if $p\not=m$ and, otherwise, the
vertex $v(\{z_1,z_2,\dots, z_p=z_m\})$. Hence, if $p\not=m$ then 
$z_j\in V(N')$, $2\leq j\leq p-1$, and otherwise, $z_j\in V(N')$ with
$j=1,\dots,m$. Consequently, $\underline{N'}$ contains 
two cycles that intersect in a path of length 1 or more
in each case. Since, by construction, each cycle is the union of 
two directed paths in $N$ that have the same start and end vertex 
this implies
that $N'$ is not of the specified form, a contradiction.

Now, if there exist just two leaves $l_{i_1}$ and $l_{i_2}$ 
in $\sigma$ that are distinct,
then choose some $l\in L(N)-\{l_{i_1},l_{i_2}\}$, which must exist
as $|X|\geq 3$. Since each of $l_{i_1}$ 
and $l_{i_2}$ is reachable from 
$v(\{z_1,z_2,\dots, z_m\})\in V(N)$ it follows that 
$v(\{z_1,z_2,\dots, z_m\})$
is a vertex in the trinet $N'$ on $\{l, l_{i_1},l_{i_2}\}$ displayed by $N$.
But then $z_j\in V(N')$, $1\leq j\leq m$ which implies that
$\underline{N'}$ contains two cycles that intersect in a path of
length at least 1. As before, this yields a contradiction.

So suppose that $l_i=l_j$ for all $i,j\in\{1,\dots, m\}$. Let 
$L_{w(l_1)}\subseteq L(N)$ denote the set of leaves of $N$ that
are reachable from $w(l_1)$.  We claim that $w(l_1)$ 
is not a cut vertex  of $N$.
Suppose for contradiction that $w(l_1)$ is 
a cut vertex of $N$. Then, since $N$ is recoverable, there must exist a leaf 
$l\in L(N)-L_{w(l_1)}$ that is reachable from $\rho_N$ without
 crossing $w(l_1)$.
Choose some $l'\in L(N)-\{l_1,l\}$, which must exist as $|X|\geq 3$. 
Since $l_1$ is reachable from each
of $z_j$, $1\leq j\leq m$, $v(\{z_1,z_2,\dots, z_m\})\in V(N)$ 
must be a vertex in the trinet $N'$ on $\{l,l',l_1\}$
displayed by $N$. But then $z_i\in V(N')$, $1\leq i\leq m$,
and so we obtain a contradiction as before.
Thus, $w(l_1)$ cannot be a cut vertex of $N$, as claimed.

Thus, there must exist some leaf $l\in L_{w(l_1)}$ that is reachable from
$\rho_N$ without crossing $w(l_1)$. But then 
$l_1\not=l$, by the definition of $w(l_1)$. 
Arguments similar to the ones used in the 
previous case can be now used to 
obtain a final contradiction. Thus, $|Z(C)|=1$
must hold for every cycle $C$ of $\underline{N}$.
\end{proof}

To establish Theorem~\ref{green} we will
use one further result that follows from the last proposition.
Suppose $N$ is a phylogenetic network on $X$, $|X|\geq 3$, and $C$
is a cycle in $\underline{N}$
with $|Z(C)|=1$. Then we denote the unique vertex in $C$ that 
has two of its incoming arcs contained in $A(C)$ by $h_{C}$. 

\begin{corollary}\label{not-cut-vertex}
Let $N$ be a recoverable phylogenetic network on $X$, $|X|\geq 3$,
such that every trinet in $Tr(N)$ is isomorphic to one of the
fourteen trinets on $\{x,y,z\}$ depicted in Fig.~\ref{trinets}.
Let $C_1$ and $C_2$ denote two distinct cycles of $\underline{N}$ for which 
$A(C_1)\cap A(C_2)\not=\emptyset$ holds,
and let $l\in L(N)$ denote a leaf of $N$ that is reachable from 
both $h_{C_1}$ and $h_{C_2}$. Then $w(l)$ is not a cut vertex of $N$.
\end{corollary}
\begin{proof} 
Suppose for contradiction that this is not the case, that is, there exists
a recoverable phylogenetic network $N$ on $X$, two distinct
cycles $C_1$ and $C_2$ in $\underline{N}$ 
with $A(C_1)\cap A(C_2)\not=\emptyset$,
and a leaf $l\in L(N)$ of $N$ that is reachable from $h_{C_1}$ and from
$h_{C_2}$ but that $w(l)$ is a cut vertex of $N$. Since $N$ is recoverable,
there must exist a leaf $l'\in L(N)-\{l\}$ of $N$ that 
is reachable from $\rho_N$ without crossing $w(l)$. 

Now let $z_i\in V(N)$ denote the unique vertex in $Z(C_i)$ $i=1,2$. 
Note that $z_1=z_2$ might hold.
Since $l$ is clearly also reachable from both $z_1$ and $z_2$, there
must exist a directed path from $\rho_N$ to $v(\{z_1,z_2\})$ that 
crosses $v(\{l,l'\})$. Choose some $l''\in L(N)-\{l,l'\}$
which must exist as $|X|\geq 3$. 
Then the trinet $N'$ on $\{l,l',l''\}$ displayed by $N$ 
contains the vertex $v(\{z_1,z_2\})$ and thus
every arc in  $A(C_1)\cup A(C_2)$.
Since $A(C_1)\cap A(C_2)\not =\emptyset$ it follows that
 $N'$ is not of the specified form which is impossible.
\end{proof}

We now prove the main result of this section:

\begin{theorem}\label{green}
Suppose that $N$ is a recoverable phylogenetic network on $X$, $|X|\geq 3$.
Then $N$ is 1-nested if and only if
every trinet in $Tr(N)$ is isomorphic to one of the 
fourteen trinets depicted in Fig.~\ref{trinets}. 
\end{theorem}
\begin{proof} 
If $N$ is 1-nested then, by Lemma~\ref{1-nested-implies trinets},
the trinets in $Tr(N)$ are of the specified form.

Conversely, suppose that the trinets in $Tr(N)$ are of
the specified form. 
Assume for contradiction that $N$ is not $1$-nested. Then there 
must exist two cycles $C_1$ and $C_2$ in $\underline{N}$ which intersect in 
more than one vertex. Moreover, amongst all such pairs of cycles, there  
must exist a pair $C_1$ and $C_2$ for which the following holds:
There is a path $P$  with $V(P)\subseteq C_1\cap C_2$ which 
has an end vertex $x_2\in V(P)$ such that the edge $\{x_1,x_2\}\in E(P)$
is the arc $(x_1,x_2)$ in $ A(N)$ and 
$\{y,x_2\}\not\in E(C_1)\cap E(C_2)$, for all 
$y\in (C_1\cap C_2)-\{x_1,x_2\}$.
Choose some $z_i\in Z(C_i)$, $i=1,2$ and note that, by
Proposition~\ref{cycles}, $|Z(C_i)|=1$. However note that
$z_1=z_2$ might hold.

Let $l_i\in L(N)$ denote a leaf of $N$ that is reachable from $h_i=h_{C_i}$, 
$i=1,2$. Then one of the three generic cases (a) - (c)
pictured in Fig.~\ref{cases} must hold.
Note that in the case of (b) and (c) we can choose $l_1$ to equal
$l_2$ since in case of (b) we have $x_2=h_2$
and in case of (c) we have $x_2=h_2=h_1$.
\begin{figure}[t] \centering
\includegraphics[scale=0.5]{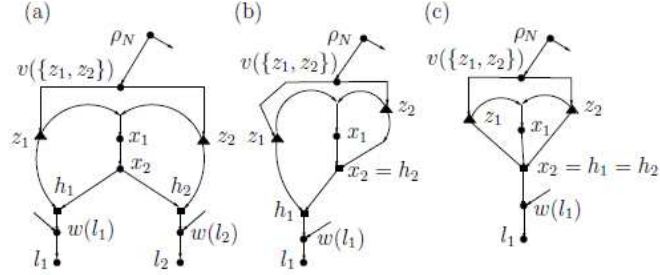}
\caption{The three generic cases considered in 
the proof of Theorem~\ref{green}.
The vertices in $C_i$, $i=1,2$ which have two of their incoming (outgoing) 
arcs contained in $A(C_i)$, $i=1,2$ are marked with squares (triangles). The 
leaves of $N$ plus the vertices $\rho_N$, $x_1$, $x_2$, $v(\{z_1,z_2\})$ 
and $w(l_i)$, $i=1,2$ are marked by dots.
For clarity all other vertices are not marked.
The directed lines represent directed paths rather than arcs.}
\label{cases}
\end{figure}

Suppose first that Case (a) holds. We begin 
by considering the case $l_1=l_2$.
Since $N$ is recoverable, Corollary~\ref{not-cut-vertex} implies that
$w(l_1)$ is a not a cut vertex of
$N$. Let $L_{w(l_1)}\subseteq L(N)$ 
denote the set of leaves of $N$ that are reachable from $w(l_1)$.
Then there must exist a leaf $l\in L_{w(l_1)}$ of $N$ that is reachable 
from $\rho_N$ without crossing $w(l_1)$. By the definition of $w(l_1)$,
$l_1\not =l$. Since $l$ is reachable from
$z_1$ and from $z_2$, there must exist a directed path from
$\rho_N$ to $v(\{z_1,z_2\})$ that crosses $v(\{l_1,l\})$.
Choose some $l'\in L(N)-\{l_1,l\}$, which must exist as
$|X|\geq 3$. Then the trinet $N'$ on $\{l_1,l,l'\}$ 
displayed by $N$ contains the vertex
$v(\{z_1,z_2\})$. Thus, $\underline{N'}$ contains
two cycles that intersect in the 
edge $\{x_1,x_2\}$. Since each cycle is the union of two directed
paths in $N$ that have the same start vertex and the same end vertex,
it follows that $N'$ is not of the specified 
form, a contradiction. Thus $l_1\not=l_2$ must hold.

Since $|X|\geq 3$, we may choose some $l\in L(N)-\{l_1,l_2\}$.
But then similar arguments applied to
the trinet $N'$ on $\{l_1,l_2,l\}$ 
displayed by $N$ yields a contradiction.

Similar arguments can be used to show that 
Case (b) and Case (c) lead to a contradiction.
But this implies that there cannot exist
two distinct cycles of $\underline{N}$
that intersect in more than one vertex. 
Thus, $N$ must be 1-nested.
\end{proof} 

As a corollary we see that if 
all of the trinets displayed by a 
recoverable phylogenetic network are trees
then the network must be a tree.

\begin{corollary}
Suppose $N$ is a recoverable phylogenetic network on $X$, $|X|\geq 3$.
Then $N$ is a phylogenetic tree on $X$ if and only if every trinet
in $Tr(N)$ is isomorphic to either the trinet $T_1(x,y,z)$ or the
trinet $T_2(x,y,z)$ on $\{x,y,z\}$.
\end{corollary}

\begin{proof}
This is an immediate consequence of 
Theorem~\ref{green} and the fact that
if $N$ is  recoverable 1-nested network on $X$ then
$\underline{N}$ contains a cycle if and only if 
there exists a trinet $N'\in Tr(N)$ such that $\underline{N'}$ 
contains a cycle.
\end{proof}

\section{Cherries, cactuses and reductions}\label{cherries-catuses-etc}

In the next section, we shall show that the set of trinets displayed
by a phylogenetic network encode the network. To
do this, we will use some operations that can be performed on 
1-nested networks to produce new 1-nested networks
which we shall now introduce. These operations 
are very closely related to the ``$R,T$ and $G$--operations'' 
presented in \cite[Section 4]{CLR11}. In consequence, we
shall omit the proofs of the results that
we state concerning our operations, instead 
citing the related results in \cite[Section 4]{CLR11}
which have very similar proofs.

Suppose $N=(V,A)$ is a 1-nested network on $X$, $|X| \ge 2$.
We call a subset $S \subseteq X$
a {\em cherry} of $N$ if $|S| \ge 2$ and there is some $v_S \in V$
such that $(v_S,x) \in A$ for all $x \in S$ and $(v_S,x) \notin A$ for 
all $x \in X-S$ (see Fig.~\ref{cactuspicture}(a)). 
Moreover,  we shall call such a cherry {\em isolated} 
if $outdegree(v_S)=|S|$ 
and $indegree(v_S)=1$ (see Fig.~\ref{cactuspicture}(b)).
Note that if $S$ is a cherry of $N$ and $S=X$, 
then $N$ is isomorphic to a bush on $X$.
\begin{figure}[b]
\includegraphics[scale=0.5]{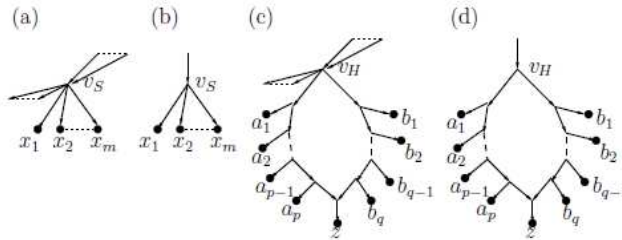}
\caption{(a) A cherry $S=\{x_1,x_2,\dots,x_m\}$, $m\ge 2$,
(b) an isolated
cherry $S=\{x_1,x_2,\dots,x_m\}$, $m\ge 2$, (c) a cactus 
$H = (a_1,a_2,\dots,a_p:b_1,b_2,\dots,b_q:z)$, $p \ge 1, q \ge 0$,
and (d) an isolated cactus 
$H = (a_1,a_2,\dots,a_p:b_1,b_2,\dots,b_q:z)$, $p \ge 1, q \ge 0$.
Note that the arcs ending at $v_S$ and $v_H$ in 
(a) and (c) do not necessarily exist.}
\label{cactuspicture}
\end{figure}
We now define a related concept. 
If $|X| \ge 2$, we call a tuple 
$H = (a_1,a_2,\dots,a_p:b_1,b_2,\dots,b_q:z)$ of distinct elements of $X$
with $p \ge 1$, $q \ge 0$ a {\em cactus} of $N$ (with {\em support} 
$S = \{a_1,a_2,\dots,a_p, b_1,b_2,\dots,b_q,z\}$) if 
there is cycle $C_H$ in $\underline{N}$ 
with split vertex $v_H$ such that 
the network induced by $N$ on $C_H \cup S$
is as pictured in Fig.~\ref{cactuspicture}(c) 
(note that if $q=0$, we
take the tuple to be $H = (a_1,a_2,\dots,a_p:\emptyset:z)$).
Moreover, such a cactus $H$ is called {\em isolated} 
if $indegree(v_H)= 1$ and   
$outdegree(v_H)=2$ (see  Fig.~\ref{cactuspicture}(d)).
Note that a two-leafed network on a set
of size two is a cactus.

Now, suppose that $N$ is 1-nested network on $X$, $|X|\ge 2$. 
In case there is a non-isolated cherry $S$ of $N$ and $z \in S$,
then we define a {\em cherry reduction}
$C=C_{z:S}$ on $N$ to be the network $C_{z:S}(N)$
which is obtained 
by removing all leaves in $S$ except $z$ from $N$, 
together with their incident arcs.
In addition, if $S$ is an isolated cherry 
of $N$ and $z \in S$, then we define an {\em isolated cherry reduction}
$\overline{C}=\overline{C}_{z:S}$ on $N$ to be the 
network $\overline{C}_{z:S}(N)$ which is obtained 
by removing all leaves in $S$ from $N$, together 
with their incident arcs, 
and replacing the vertex $v_S$ by $z$, 
which now becomes a leaf of the new network.

Similarly, suppose there is a cactus
$H = (a_1,a_2,\dots,a_p:b_1,b_2,\dots,b_q:z)$ of $N$
with support $S$.
If $H$ is not isolated, then we define a {\em cactus reduction}
$H = H_{z:S} = 
H_{a_1,a_2,\dots,a_p:b_1,b_2,\dots,b_q:z}$ on $N$ to be the network 
$H_{a_1,a_2,\dots,a_p:b_1,b_2,\dots,b_q:z}(N)$ which is obtained 
by removing the vertices $(C_H - \{v_H\}) \cup (S-\{z\})$, together 
with their induced arcs plus the two outgoing arcs of $v_H$ contained in
$A(C)$, from $N$ and then adding 
in the new arc $(v_H,z)$.
In addition, if $H$ is isolated, then 
we define an {\em isolated cactus reduction}
$\overline{H} = \overline{H}_{z:S} = 
\overline{H}_{a_1,a_2,\dots,a_p:b_1,b_2,\dots,b_q:z}$ 
on $N$ to be the network 
$\overline{H}_{a_1,a_2,\dots,a_p:b_1,b_2,\dots,b_q:z}(N)$ 
which is obtained 
by removing the vertices $(C_H - \{v_H\}) \cup (S-\{z\})$, together 
with their induced arcs
plus the two outgoing arcs of $v_H$, from $N$ and replacing $v_H$ with $z$.

It is clear that  the networks $C_{z:S}(N)$, $\overline{C}_{z:S}(N)$, 
$H_{a_1,a_2,\dots,a_p:b_1,b_2,\dots,b_q:z}(N)$ 
and ${\overline H}_{a_1,a_2,\dots,a_p:b_1,b_2,\dots,b_q:z}(N)$
are all 1-nested networks on the set $X - (S - \{z\})$ and that they
all have $|S|-1$ less leaves than $N$. Moreover we have:

\begin{proposition} \label{cherryorcactus}
\cite[Proposition 2]{CLR11}
Suppose that $N$ is a 1-nested network on $X$, $|X| \ge 1$.
If $|X|\ge 2$, then at least one of the reductions 
$C$, $\overline{C}$, $H$, $\overline{H}$ may be applied to $N$.
Moreover, if none of the reductions $C$, $\overline{C}$, 
$H$, $\overline{H}$ may be applied to $N$, then $|X|=1$
and $N$ is the bush on $X$.
\end{proposition}

We can also define `inverses' of $C^{-1}, \overline{C}^{-1}, 
H^{-1}, \overline{H}^{-1}$ of
the reductions $C$, $\overline{C}$, $H$, $\overline{H}$ as follows.
Given a 1-nested network $N$ on $X$, $|X| \ge 1$, a leaf $z \in X$ of $N$, 
and a set finite $S$ with $|S| \ge 2$ and $S \cap X =\{z\}$, we define the 
{\em cherry expansion}  $C^{-1}_{z:S}$ of $N$ to be 
the network $C^{-1}_{z:S}(N)$  obtained
by replacing leaf $z$ by a new vertex $v$, and 
adding in new arcs $(v,s)$ for all $s \in S$.
Clearly $C^{-1}_{z:S}(N)$ is a 1-nested network on $X \cup S$.
Isolated cherry, cactus and isolated cactus 
expansions $\overline{C}^{-1}$, 
$H^{-1}$, $\overline{H}^{-1}$, 
corresponding to $\overline{C}$, $H$ and 
$\overline{H}$, are defined in a similar way.

It is straight-forward to see that a reduction 
and its corresponding expansion are mutual inverses, 
in that when one is applied to a 1-nested network $N$ 
on $X$ and then its inverse, 
we obtain a network that is isomorphic to $N$. 
Moreover, we have:

\begin{lemma}\label{isomorphic}
\cite[Lemma 4]{CLR11}
Let $N$ and $N'$ be two 1-nested networks on $X$, $|X| \ge 3$.
If $N$ and $N'$ are isomorphic, then if one of the 
reductions $C$, $\overline{C}$, $H$, 
$\overline{H}$ (respectively, expansions $C^{-1}$, 
$\overline{C}^{-1}$, $H^{-1}$, $\overline{H}^{-1}$)
may be applied to $N$, then the same one
may also be applied to $N'$ and the
two resulting 1-nested networks are isomorphic.
\end{lemma}

\section{Encoding 1-nested networks with trinets}\label{encoding-1-nested}

In this section we show that the set of (necessarily 1-nested) trinets 
displayed by a 1-nested network $N$ on $X$ encodes $N$
(see Theorem~\ref{encode}).

We begin by describing how to characterize 
cherries and cactuses in a 1-nested 
network in terms of their trinets, starting with cherries. 
To this end,
we associate to a trinet set $\T$ on $X$ and a non-empty
subset $S\subseteq X$ the trinet set
$$
\T|_{S} := \{ N \in \T \,:\, S \cap L(N) \neq \emptyset  \}.
$$

\begin{lemma} \label{fatcherry}
Suppose $N=(V,A)$ is a 1-nested network on $X$, $|X| \ge 3$, 
and let $S\subseteq X$ with $|S|\geq 2$. Let $\T$ be a non-empty 
subset of $Tr(N)$. Then $S$ is a cherry of $N$ with 
$\T =Tr(N)|_S$ if and only if $\T$ satisfies 
the following properties:
\begin{itemize}
\item[(C1)] $L(N')\cap S\not=\emptyset$, for all $N'\in\T$
(or equivalently, $\T|_S=\T$).
\item[(C2)] For all $\{x,y\} \in {S \choose 2}$ and all $z \in X-S$, 
either $T_1(x,y,z)$, $T_2(x,y,z)$, $N_3(z,x,y)$, $N_4(x,y,z)$, 
$N_9(x,y,z)$ or $N_{10}(z,x,y)$ is in $\T$.
\item [(C3)] For all  $\{x,y,z\} \in  {S \choose 3}$, $T_2(x,y,z) \in \T$.
\item [(C4)] There is no $S' \subseteq X$ such 
that $S \subset S'$ and $\T$
satisfies (C2) and (C3) with $S$ replaced by $S'$.
\end{itemize}
Moreover, if this is the case and 
$S \neq X$ (or, equivalently, $|X-S| \ge 1$), then 
$S$ is isolated if and only if $\T$ also satisfies:
\begin{itemize}
\item[(C5)] For all $\{x,y\} \in {S \choose 2}$ and 
all $z \in X-S$, either 
$T_1(x,y,z)$, $N_3(z,x,y)$ or $N_4(x,y,z)$ is  contained in $\T$.
\end{itemize}
\end{lemma}
\begin{proof}
Suppose $\T =Tr(N)|_S$ holds for some cherry $S$ of $N$. 
Then it is straight-forward to check that $\T$ satisfies (C1)--(C4).

Conversely, suppose $\T$ satisfies (C1)--(C4). Let $v = v(S)$. 
Note that $v(\{x,y\})=v$ for all $\{x,y\} \in {S \choose 2}$, 
since otherwise there would 
exist some $z \in S$ such that $T_2(x,y,z) \not\in \T$, 
in contradiction to (C3). 
Moreover, suppose there were some $z \in X-S$, $x \in S$ with 
$v(\{z,x\}) >_{N} v$. Let $y \in S -\{x\}$ (which 
exists since $|S|\ge 2$). 
Then none of the trinets $T_1(x,y,z)$, $T_2(x,y,z)$,
$N_3(z,x,y)$, $N_4(x,y,z)$
$N_9(x,y,z)$ or $N_{10}(z,x,y)$ could be contained in $\T$, in 
contradiction to (C2). Thus, for all $z\in X-S$ and all
$x\in S$, we have $v(\{z,x\}) <_{N} v$ with possibly
equality holding. It follows that
$(v,x) \in A$ for all $x \in S$. 

Now, suppose there is some $r \in X-S$ with $(v,r) \in A$. 
Let $S' = S \cup \{r\}$.  Then it is straight-forward
to check that, for all $x \in S'$ and all $z \in X- S'$, 
either $T_1(x,r,z)$, $T_2(x,r,z)$, $N_3(z,x,r)$, $N_4(x,r,z)$, 
$N_9(x,r,z)$ or $N_{10}(z,x,r)$ is in $\T$, and that 
$T_2(x,y,r) \in \T$ for all $\{x,y\} \in {S' \choose 2}$. This implies
that $S'$ satisfies (C2) and (C3) with $S$ replaced by $S'$, which 
contradicts (C4). In particular, it follows that $S$ is a cherry of $N$.

To see that $\T= Tr(N)|_S$ holds note first that
$\T\subseteq Tr(N)|_S$ is a consequence of (C1). To
see that $Tr(N)|_S\subseteq \T$ suppose $N'\in Tr(N)|_S$.
Then $L(N')\cap S\not =\emptyset$ and so  $N'\in\T$
follows from considering
the size of the intersection $L(N')\cap S$ in conjunction 
with Properties (C2) and (C3). 
 
To complete the proof, suppose that $X\not =S$. First note that if 
$S$ is an isolated cherry of $N$, then (C5) clearly holds.
Conversely, if $\T$ satisfies (C5), then let $v\in V$ be the vertex
with $(v,x) \in A$ for all $x \in S$ and $(v,x)\not\in A$ 
for all $x \in X-S$
(which exists since $S$ is a cherry by (C2)--(C4)).
Then $outdegree(v)= |S|$, since otherwise 
there would exist some $\{x,y\} \in {S \choose 2}$ and $z \in X-S$ 
with $z >_N v$ such that either $T_2(x,y,z)$ or $N_{10}(z,x,y) \in \T$, 
in contradiction to (C5).

Now, since $|X-S|\ge 1$, $indegree(v) \ge 1$. Suppose 
$indegree(v) > 1$. Then there must exist some $z \in X-S$
and $\{x,y\} \in {S\choose 2}$ such that $N_9(x,y,z) \in \T$, which
contradicts (C5). Therefore $indegree(v) = 1$, which completes the proof.
\end{proof}

We now present a similar result for cactuses. 

\begin{lemma} \label{cactus}
Let $N$ be a 1-nested network on $X$, $|X| \ge 3$,
and let $H=(a_1,\dots, a_p:b_1,\dots,b_q:z)$ be a tuple of distinct
elements in $X$ with $p\geq 1$ and $q\geq 0$. Put
$S=\{a_1,\dots, a_p, b_1,\dots,b_q, z\}$ and let 
$\T$ be a non-empty subset of 
$Tr(N)$. Then $H$ is a cactus of $N$ with support 
$S$ and  $\T =Tr(N)|_S$
if and only if, with $A=\{a_1,\dots,a_p\}$ and $B=\{b_1,\dots,b_q\}$,
$\T$ satisfies the following properties:
\begin{itemize}
\item[(H1)]  $L(N')\cap S\not=\emptyset$
for all $N'\in \T$ (or, equivalently $\T|_S=\T$).
\item[(H2)] $N_1(x,z,y) \in \T$ for all $x \in A$, $y \in B$.
\item[(H3)] $N_2(z,x,x') \in \T$ for all 
$x=a_i, x'=a_j$, $1 \le i < j \le p$, or 
$x=b_r,x'=b_s$, $1 \le r < s \le q$.
\item[(H4)] $T_1(x,x',x'') \in \T$ for all 
$x=a_i, x'=a_j, x''=a_k$, $1 \le i < j < k \le p$,
or  $x=b_r, x'=b_s, x''=b_t$, $1 \le r < s < t \le q$.
\item[(H5)] For all $w \in X-S$ either $N_5(z,x,w)$, or 
$N_6(z,x,w)$, or $N_7(w,x,z) $, or $N_8(z,x,w)$, or 
$N_{11}(z,x,w)$, 
or $N_{12}(w,x,z)$ is contained in $\T$, for all $x \in A$ or $x \in B$.
\item[(H6)] For all $w \in X-S$,  either $T_1(x,x',w)$, or 
$N_3(w,x,x')$, or $N_4(x,x',w)$ is contained in $\T$,
for all $x\neq x' \in A$ or $x \neq x'\in B$.
\item[(H7)] For all $w \in X-S$,  one of 
$T_1(x,y,w)$, $T_1(x,w,y)$, $T_1(y,w,x)$, $T_2(x,y,w)$, $N_4(x,y,w)$,
and $N_1(x,w,y)$ is contained in $\T$, for all $x \in A$ and $y \in B$.
\item[(H8)] There exists no tuple 
$H=(c_1,\dots, c_t: d_1\dots,d_s:z)$ of distinct elements in $X$,
$t\geq 1$ and  $s\geq 0$, with 
$S\subsetneq S':=\{c_1,\dots, c_t, d_1\dots,d_s,z\}$ such 
that $\T$ satisfies (H2)--(H7) for $S'$.
\end{itemize}
Moreover, if this is the case and $S\not =X$, then 
$H$ is isolated if and only if $\T$ also satisfies: 
\begin{itemize}
\item[(H9)] For all $w \in X-S$, $T_2(x,y,w) \not\in \T$,
for all $x \in A$, $y \in B$, and
$N_8(z,x,w) \not\in \T$, $N_{11}(z,x,w) \not\in \T$ and 
$N_{12}(w,x,z) \not\in \T$
for all $x \in A$ or $x \in B$.
\end{itemize}
\end{lemma}
\begin{proof}
Suppose $H$ is a cactus of $N$ with support $S$ and $\T=Tr(N)|_S$.
Then it is straight-forward to see that $\T$ must satisfy (H1)--(H8).

Conversely, suppose  
$\T$ satisfies (H1)--(H8) with $A$ and $B$
as specified. We claim that $H$ is 
a cactus of $N$ with support $S$. We prove the claim 
for $q=0$ and remark that
the proof for $q\geq 1$ is similar. Let $x\in A$. 
If $|S|=2$ then choose some $w\in X-S$.
By (H5), one of the trinets
$N_5(z,x,w)$, $N_6(z,x,w)$, $N_7(w,x,z)$, $N_8(z,x,w)$, 
$N_{11}(z,x,w)$, $N_{12}(w,x,z)$  must be contained in $\T$.
But then $H$ must clearly be a cactus of $N$ (with support $S$). 
 
Assume that $ |S|\geq 3$. Then $|A|\geq 2$ and, by (H3), 
$N_2(z,a_i,a_j)\in\T$  or $N_2(z,a_j,a_i)\in\T$ holds for all 
$\{i,j\}\in{\{1,\dots, p\}\choose 2}$. Since
$\T\subseteq Tr(N)$, there must exist a cycle $C_{i,j}$
in $\underline N$ with split vertex $v_{i,j}:=v_{C_{i,j}}$ and end
vertex $b_{i,j}:=b_{C_{i,j}}$ that gives rise to that trinet
on $\{z,a_i,a_j\}$, 
$\{i,j\}\in{\{1,\dots, p\}\choose 2}$. We show that $C_{i,j}=C_{k,l}$ 
holds for all $\{i,j\},\{k,l\} \in{\{1,\dots, p\}\choose 2}$. To see this
it suffices to show that $C_{i,j}=C_{i,l}$ holds for all 
$i\in\{1,\dots, p\}$ and all $\{k,l\} \in{\{1,\dots, p\}-\{i\}\choose 2}$. 
So assume 
for contradiction that there exists some $i\in\{1,\dots, p\}$ and some 
$\{j,l\} \in{\{1,\dots, p\}-\{i\}\choose 2}$ with $C_{i,j}\not=C_{i,l}$. 
Without loss of generality assume that $i=1$. 

Note that since $N$ is 1-nested there must exist, 
for all $t\in\{2,\dots,p\}$ and 
all $x\in\{a_1,a_t,z\}$, a unique last vertex $v_x^{1,t}$ in $C_{1,t}$ 
that lies on every
path from $\rho_N$ to $x$. Clearly, $v_z^{1,t}$ is the
end vertex of $C_{1,t}$  and $v_{a_1}^{1,t}$ is neither  
the end vertex nor the split vertex of $C_{1,t}$, $t\in\{2,\dots,p\}$. 
Put $v=v(\{v_{1,j},v_{1,l}\})$. 

We first show that $b_{1,j}=b_{1,l}$. Suppose  for contradiction that
$b_{1,j}\not=b_{1,l}$. Then since $indegree(z)=1$, 
there must exist a vertex $y_z$ distinct from $z$ that lies simultaneously 
on any path in $N$ from $b_{1,j}=v_{z}^{1,j}$ to 
$z$ and on  any path in $N$  from $b_{1,l}=v_{z}^{1,l}$ to $z$. 
Without loss of 
generality, we may assume that $y_z$ is as close to $z$ as 
possible. So there must exist a cycle $C$ in $\underline{N}$ 
with $\{v, v_{1,j}, b_{1,j},y_z, b_{1,l}, v_{1,l}\}\subseteq C$
with possibly $v=v_{1,j}$ or $v=v_{1,l}$ or $v=v_{1,j}=v_{1,l}$
or $b_{1,j}=y_z$ or $b_{1,l}=y_z$ holding. Since
$\{v_{1,j},b_{1,j}\}\subseteq C\cap C_{1,j}$
and $N$ is 1-nested this is impossible. 
Thus $b_{1,j}=b_{1,l}$, as required. 
 
Similar arguments with $z$
replaced by $a_1$ in the definition of $y_z$ also imply that
$v_{a_1}^{1,j}=v_{a_1}^{1,l}$ must hold. But then 
$C_{1,j}$ and $C_{1,l}$ intersect in more than one
vertex which is impossible as $N$ is 1-nested. Thus 
$C_{1,j}=C_{1,l}$ must hold for all $j,l\in \{2,\dots,p\}$.
Moreover, by (H4), $v_{a_r}^{1,2}\not=v_{a_s}^{1,2}$ for all 
$\{r,s\}\in{\{1,\dots, p\}\choose 2}$. Thus there exists a directed
path $P$ from $v_{1,2}$ to $b_{1,2}$ that crosses the vertices 
$v_{a_1}^{1,2}, v_{a_2}^{1,2},\dots, v_{a_p}^{1,2}$ in that order.

To finish the proof of the claim that $H$ is a cactus 
of $N$ with support $S$, we next establish that 
$V(P)=Y:=\{v_{a_1}^{1,2}, v_{a_2}^{1,2},\dots, v_{a_p}^{1,2}, 
v_{1,2},b_{1,2}\}$. Suppose for contradiction that this is
not the case and that there exists some $u\in V(P)-Y$. 
Without loss of generality, we may assume that $(u,v_{a_1}^{1,2})\in A(P)$. 
Since $N$ is 1-nested, there exists some leaf $w\in L(N)-S$ that
is reachable  from $u$ without crossing any further vertex in $C_{1,2}$.
We distinguish the cases that $X=S$ and that $X\not=S$.
If $X=S$ then this is impossible and so $V(P)=Y$, as required.
Since $C_{1,2}$ is a cycle in $\underline{N}$ and
$N$ is 1-nested it follows that $(v_{1,2},b_{1,2})$ is an arc in
$N$. But this implies that $H$ is a cactus of $N$ (with support $S$).

So assume that $S\not=X$. Then (H6) applied to $a_1$, $a_2$, and $w$,
combined with the fact that $N$ is 1-nested, implies that the trinet
$T_1(a_1,a_2,w)$ is contained in $\T$. 
But then $\T$ satisfies (H2)--(H7) for the
support $S\cup\{w\}$ of the tuple 
$H'=(w,a_1,\dots, a_p:\emptyset: z)$. In view of (H8), this is
impossible. Thus, $V(P)=Y$, as required. 

We now show that
$(v_{1,2},b_{1,2})\in A(C_{1,2})$. Suppose this is not the case and there
exists some $u\in C_{1,2}-V(P)$. Without loss of generality we may assume
$(v_{1,2},u)\in A(C_{1,2})$. Then there exists a leaf $w\in L(N)-S$ such 
that $u$ is the last vertex in $C_{1,2}$ on any path from $\rho_N$ to $w$.
But then the trinet on $\{w, a_1, z\}$ is not as specified in (H5)
which is impossible. Thus, $(v_{1,2},b_{1,2})\in A(C_{1,2})$, as required.
It follows that $H$ must be 
a cactus of $N$ (with support $S$) in this case, too.

To see that $Tr(N)|_S=\T$, let $N'\in Tr(N)|_S$. Then 
$L(N')\cap S\not =\emptyset$. By distinguishing the
cases that $|L(N')\cap S| =1,2$, or $3$, it is straight forward
to show that $N'\in \T$ using Properties (H2)--(H8). 
Also $\T\subseteq Tr(N)|_S$ holds by Property (H1).

It remains to show that if $H$ is a cactus of $N$ with support $S$ 
and $S\not=X$ then
$H$ is isolated if and only if $\T$ satisfies (H9).  Assume that $H$
is a cactus  of $N$ with support $S$ and that $S\not=X$. Then it is 
straight forward to check that $\T$ satisfies (H9).  

Conversely, 
assume that $\T$ satisfies (H9).
We need to show that $outdegree(v_H)=2$ and that $indegree(v_H)=1$. 
We again prove the case $q=0$ and remark that
the arguments for $q\geq 1$ are similar.
Since $H$ is a cactus of $N$ we clearly have
$outdegree(v_H)\geq2$. Assume for contradiction that
$outdegree(v_H)>2$. Then since $N$ is 1-nested and $X\not=S$
there must exist some $w\in X-S$ that is reachable from
$v_H$ without crossing a vertex in $C_H-\{v_H\}$, where 
$C_H$ is the cycle in $\underline{N}$
corresponding to $H$. But then there exists some $x\in A$ 
such that $N_8(z,x,w)$ or $N_{12}(w,x,z)$ is 
contained in $\T$ contradicting (H9). Thus, $outdegree(v_H)=2$,
as required. But then  $indegree(v_H)\geq 1$
as $S\not=X$.
Assume for contradiction that $indegree(v_H)>1$.
Then since $S\not= X$ there must exist some $w\in X-S$ such that
$N_{11}(z,a,w)\in \T$ for some $a\in A$ contradicting again (H9).
Thus, $indegree(v_H)=1$. This completes the 
proof of Lemma~\ref{cactus}.
\end{proof}

Now, let $R_{z:S}$ ($R_{z:S}^{-1}$) 
denote any of the four reductions $C$, 
$\overline{C}$, $H$, $\overline{H}$
with $z$ and $S$ as specified in the definition
of the reductions. Then
it is straight-forward to check that if $N$ is a
1-nested network on $X$, $|X|\geq 3$, and 
$$
\T_{z:S} = \{ N' \in Tr(N) \,:\, S \cap L(N') \neq 
\emptyset \mbox{ and } S \cap L(N') \neq \{z\} \},
$$ 
then
\begin{equation}\label{identity}
Tr(N) = Tr(R_{z:S}(N))  \amalg \T_{z:S},
\end{equation}
or, in other words, $Tr(R_{z:S}(N)) = Tr(N) - \T_{z:S}$.

\begin{theorem} \label{encode}
Suppose that $N$ and $N'$ are both 1-nested networks on $X$, $|X| \ge 3$.
Then $Tr(N) = Tr(N')$ if and only if $N$ is isomorphic to $N'$.
\end{theorem}

\begin{proof}
Suppose first that $N$ is isomorphic to $N'$. 
Then $Tr(N) = Tr(N')$ follows 
immediately by using induction on $|X|$, 
Lemma~\ref{isomorphic} and (\ref{identity}).

To prove the converse we also use induction on $|X|$. 
If $|X|=  3$, then the converse obviously holds. 
So, suppose that, for all $1 \le |X|\le m$, $m \ge 3$, 
if $Tr(N) = Tr(N')$ then $N$ is isomorphic to $N'$. 

Let $|X|=m+1$, and suppose that $N$ and $N'$ are 1-nested 
networks on $X$ with  $Tr(N) = Tr(N')$. By 
Proposition~\ref{cherryorcactus} 
we can apply at least one of the reductions 
$R = \overline{C}, C, H, \overline{H}$ to $N$.
Therefore, since $Tr(N)=Tr(N')$, by 
Lemmas~\ref{fatcherry} and \ref{cactus}, 
we may also apply the same reduction $R$ to $N'$. 
Moreover, by (\ref{identity}) we have $Tr(R(N)) = Tr(R(N'))$. 
So, by induction, $R(N)$ is isomorphic to $R(N')$.
Therefore, by Lemma~\ref{isomorphic}, $R^{-1}(R(N))$ 
is isomorphic to $R^{-1}(R(N'))$, i.e. $N$ is 
isomorphic to $N'$, as required. 
\end{proof}

There has been some interest in 
the literature in defining metrics on
networks \cite[page 172]{HRS11}, and
various metrics have been defined for different types
of phylogenetic networks including 1-nested networks 
\cite{CRV08a,CLRV09a,CLRV09b,CLR11,CRV08b,GH11}. 
Thus the following result could
be of interest. For $X$ with $|X|\ge 3$, let ${\mathcal N}_1(X)$ 
denote the set of 1-nested networks on $X$. In addition, 
define the map
$$
d: {\mathcal N}_1(X) \times {\mathcal N}_1(X) \to \R; 
(N,N')\mapsto d(N,N') := |Tr(N) \Delta Tr(N')|,
$$
for all $N,N'\in {\mathcal N}_1(X)$.
Then the last theorem immediately implies:

\begin{corollary} \label{metric}
For $X$ with $|X|\ge 3$,  the map $d$ is a (proper)
metric on ${\mathcal N}_1(X)$.
\end{corollary}

Note that the metric $d$ can be efficiently computed
since, for $N \in {\mathcal N}_1$, it is 
possible to compute every trinet in $Tr(N)$ 
efficiently (essentially because for any $Y \in {X \choose 3}$
the vertex $v(Y)$ can be computed efficiently using, e.g.
the algorithm presented in \cite{LT79}).

\section{Constructing 1-nested networks from dense sets of 
trinets}\label{1-nested-from-dense}

In this section, we present an efficient algorithm 
which, given a dense set $\T$
of trinets, can decide whether or not it is displayed
by a 1-nested network, and if this is the case, 
constructs the network displaying $\T$
(see 
Fig.~\ref{algorithm:multicons}).

We begin by describing efficient algorithms for detecting 
cherries and cactuses.
Given a dense set $\T$ of trinets on $X$, we say 
that $S \subseteq X$, $|S|\ge 2$ 
is a {\em cherry of $\T$} if the set $\T|_S$ 
satisfies conditions (C2)--(C4)
(note that it necessarily satisfies (C1)), and
that it is {\em isolated} if it also satisfies (C5). We now show that 
cherries can be found in polynomial time in a dense set 
of trinets using the algorithm presented in Fig.~\ref{algorithm:cherry}.

\begin{lemma}\label{polyfind1}
Given a dense set $\T$  of trinets on $X$, $|X| \ge 3$, 
algorithm \textsc{FindCherry} is correct and has run-time that is 
polynomial in $|X|$.
\end{lemma}
\begin{proof}
It is straight-forward to see that algorithm \textsc{FindCherry}
has run-time that is polynomial in $|X|$.

To see that algorithm \textsc{FindCherry} is correct, first note that it will 
clearly terminate. Now, suppose that the algorithm outputs a 
(non-empty) set $S$. Then, in view of line 7, $\T|_S$ must
satisfy (C2) and (C3). Moreover, in view of the while loop (lines 6--10)
$\T|_S$ must satisfy (C4), So $S$ must be a cherry of $\T$. 
Moreover, if the output indicates that $S$ is isolated (i.e. that 
$S\not=X$ and that $\T|_S$ satisfies (C5)), then this must 
be the case in view of line 8.

Now, suppose that algorithm \textsc{FindCherry} 
outputs ``No cherry of $\T$ exists", 
and that, for the purposes of contradiction, a cherry $S$ of $\T$ 
does exist. Then, as any cherry has cardinality at least 2, if a cherry
exists then at some stage the while loop in lines 2--12 must 
encounter some $\{x,y \} \in {X \choose 2}$ with $\{x,y\}\subseteq S$.
Clearly, the algorithm will then have to output $S$, a contradiction.
Thus the algorithm \textsc{FindCherry} is correct.
\end{proof}

\begin{figure}[h]
\centering
\parbox{0cm}{\begin{tabbing}
XXX\= XX\= XX\=  XX\= XX\= XX\= XX\=  XXXXX\=  XXXXXXXXXXX\= \kill \\
{\large \textsc{FindCherry}($X$,\(\mathcal{T}\))}\\
\rule{\columnwidth}{0.5pt}\\
Input: \> \> \> A set $X$, $|X|\ge 3$, and a dense set $\T$ of trinets on $X$. \\
Output: \> \> \>A cherry $S$ of $\T$, and a boolean variable
$I \in \{\rm{T},\rm{F}\}$, with $I=\rm{T}$ \\
\>\>\> if $S$ is isolated and $I=\rm{F}$ else, or the statement 
``No cherry of $\T$ \\
\>\>\>exists".\\
\rule{\columnwidth}{0.5pt}\\
1. \> Let $S = \emptyset$, $I=\rm{F}$, $G = {X \choose 2}$.\\
2. \> While there is some $\{x,y\} \in G$ do \\
3. \>  \> If $T_1(x,y,z)$, $T_2(x,y,z)$, $N_3(z,x,y)$, $N_4(x,y,z)$, $N_9(x,y,z)$\\ 
4. \>  \> or $N_{10}(z,x,y)$ is contained in $\T$ for all $z \in X -\{x,y\}$ then do\\
5. \>  \> Let $S=\{x,y\}$, $G = \emptyset$ and $U=X-\{x,y\}$.\\
6. \>  \> While there is some $u \in U$ do \\
7. \>  \> \> If $\T|_{S \cup \{u\}}$ satisfies (C2) and (C3), then let $S = S  \cup \{u\}$.\\
8. \> \> \>   If $U = \{u\}$, $S\not=X$, and $S$ satisfies (C5), then let $I=\rm{T}$.\\
9. \> \> \> Let  $U=U-\{u\}$. \\
10. \> \> end ``do (line 6)"\\
11. \> \> else let $G = G - \{\{x,y\}\}$.\\
12. \> end ``do (line 2)"\\
13. \> If $S=\emptyset$ then output ``No cherry of $\T$ exists" else output $S$ and $I$.\\
\end{tabbing}}
\caption{Pseudo-code 
for an algorithm that either finds a cherry of a dense trinet set $\T$ 
and also checks whether it is isolated
or not or determines that no cherry of $\T$  exists.}
\label{algorithm:cherry}
\end{figure}

Now, given a dense trinet set $\T$ on $X$,
we say that a tuple $H= (a_1,a_2,\dots,a_p:b_1,b_2,\dots,b_q:z)$ of
distinct elements of $X$, $p \ge 1$, $q \ge 0$ is a {\em cactus of $\T$} 
(with support $S = A \cup B \cup \{z\}$, $A=\{a_1,\dots,a_p\}$ and
$B = \{b_1,\dots,b_q\}$)
if $\T|_S$ satisfies conditions (H2)--(H8) of Lemma~\ref{cactus}
(note that $\T|_S$ necessarily satisfies (H1)). Moreover, 
such an $H$ is {\em isolated} if  $S\not=X$ and $\T|_S$ also satisfies 
condition (H9) of Lemma~\ref{cactus}.

Note that if $H= (a_1,a_2,\dots,a_p:b_1,b_2,\dots,b_q:z)$ 
is a cactus of $\T$, then
the relation $\sim_{\T}$ defined on the set 
$Y=S-\{z\}=A \cup B$ by putting $y \sim_{\T} y'$ if and only if $y=y'$ 
or $N_2(z,y,y')$ or $N_2(z,y',y) \in \T$, for all $y,y' \in Y$, 
is an equivalence relation on $Y$ 
with (at most two) equivalence classes $A,B$. Moreover,  
the relation $<_{\T}$ defined on $Y$ by $y <_{\T} y'$ 
if and only if 
$N_2(z,y,y')  \in \T$, for all $y,y' \in Y$, is a 
strict partial order on $Y$, 
which restricts to a strict 
linear order on $A$ and also on $B$.

Using these observations, we now show that the 
algorithm presented in Fig.~\ref{algorithm:cactus} 
can be used to detect cactuses in 
a dense set of trinets in polynomial time.

\begin{lemma}\label{polyfind2}
Given a dense set $\T$  of trinets on $X$, $|X| \ge 3$, 
algorithm \textsc{FindCactus} is correct and has 
run-time that is polynomial in $|X|$.
\end{lemma}
\begin{proof}
First note that the algorithm will clearly terminate. 
Moreover, if it does output a tuple then in view of lines 12 and 13 this
must be a cactus of $\T$ and it will be isolated only if $I=\rm{T}$.
In addition, if the algorithm outputs ``No cactus of $\T$ exists", then
this must be the case. Otherwise, suppose there is some 
cactus $K=(a_1,\dots,a_p:b_1,\dots,b_q: z)$ of $\T$, $p\geq 1$, 
$q\geq 0$. Setting $A=\{a_1,\dots,a_p\}$ and $B=\{b_1,\dots,b_q\}$
it follows that $S=A\cup B\cup\{z\}$ 
is the support of $K$ and that $z$ must be at the bottom of 
some trinet in $\T$. Thus the while loop (lines 2--20)  would 
eventually find $z$ at line 3. Since $K$ is a cactus of $\T$, for 
each element $y\in Y:=A\cup B$,
 there exists some $N\in \T$ such that $y$ hangs off the
side of $N$ and $z$ is at the bottom of $N$. 
Moreover, $A$ and $B$ (in case $B\not=\emptyset$) are the 
equivalence classes of the  relation $\sim_{\mathcal T}$ defined on
$Y$ and the elements in $A$ and $B$
(again in case $B\not=\emptyset$) are strictly linearly ordered by
$<_{\mathcal T}$. Thus, the algorithm would form the
tuple $F=(a_1,\dots,a_p:b_1,\dots,b_q:z)$ (lines 10 and 11). 
Clearly, the support of $F$ is $S$. Since
$\T|_S$ satisfies (H2)--(H8) it follows that $F$ is returned by
the algorithm. However since $F=K$, this is impossible.
 
Finally, to see that algorithm 
\textsc{FindCactus} is polynomial in $|X|$, it 
is sufficient to note that lines 6--7, 8--9 and 12--13 
can all clearly be  executed in time that 
is polynomial in $|X|$.
\end{proof}

\begin{figure}[h]
\centering
\parbox{0cm}{\begin{tabbing}
XXX\= XX\= XX\=  XX\= XX\= XX\= XX\=  XXXXX\=  XXXXXXXXXXX\= \kill \\
{\large \textsc{FindCactus}($X$,\(\mathcal{T}\))}\\
\rule{\columnwidth}{0.5pt}\\
Input: \> \> \> A set $X$, $|X|\ge 3$, and a dense set $\T$ of 
trinets on $X$. \\
Output: \> \> \> A cactus $H$ of $\T$ and a boolean variable
$I \in \{\rm{T},\rm{F}\}$, with $I=\rm{T}$\\ 
 \> \> \> if $H$ is isolated and $I=\rm{F}$ else, or  
 the statement ``No cactus of $\T$ \\
 \>\>\> exists".\\
\rule{\columnwidth}{0.5pt}\\
1. \> Put $H=\emptyset$, $I=\rm{F}$, $G=X$.\\
2. \> While there is some $z \in G$ do\\
3. \>  \> If there is a trinet $N \in \T$ such that $z$ is at the 
bottom of $N$, then do\\
4. \>  \> \> Let $Y$ be the set of $y \in X-\{z\}$ such that $y$ 
hangs off the side of\\
5. \> \> \> some $N \in \T$ for which $z$ is at
the bottom of $N$. \\
6. \> \> \> If the relation $\sim_{\T}$ is an equivalence relation on $Y$ \\
7. \> \> \> that has at most two equivalence classes $E,E'$, then do\\
8.   \> \> \> \> If the relation $<_{\T}$ on $Y$ is a partial order on $Y$ 
that also restricts \\
9.   \> \> \> \> to give a strict linear order on $E$ 
and on $E'$ 
then do\\
10.   \> \> \> \> \> Let $F = (a_1,\dots,a_p:b_1,\dots,b_q:z)$ and 
$S=Y\cup\{z\}$, where \\
11. \> \> \> \> \> $E=\{a_1,\dots,a_p\}$ and $E'=\{b_1,\dots,b_q\}$ 
are ordered relative to $<_{\mathcal T}$.\\
12. \> \> \> \> \> If $\T|_S$ satisfies (H2)--(H8), then let $H=F$ and 
$G=\emptyset$ and, if \\
13. \> \> \> \> \> $\T|_S$ also satisfies (H9) then let $I=\rm{T}$, else 
let $G = G-\{z\}$.\\
14. \> \> \> \> end ``do (line 11)"\\
15.\> \> \> \> else let $G = G-\{z\}$.\\
16. \> \> \> end ``do (line 8)"\\
17. \> \> \> else let $G = G-\{z\}$.\\
18. \> \> end ``do (line 3)"\\
19. \> \> else let $G = G-\{z\}$.\\
20. \> end ``do (line 2)"\\
21. \> If $H=\emptyset$ then output ``No cactus of $\T$ exists" else 
output $H$ and $I$.\\
\end{tabbing}}
\caption{Pseudo-code for an algorithm that either finds a 
cactus of a dense trinet set $\T$ 
and also decides whether it is isolated or not or 
determines that no cactus of $\T$ exists.}
\label{algorithm:cactus}
\end{figure}

\begin{figure}[h!]
\centering
\parbox{0cm}{\begin{tabbing}
XXX\= XX\= XX\=  XX\= XX\= XX\= XX\=  XXXXX\=  XXXXXXXXXXX\= \kill \\
{\large \textsc{BuildNet}(\(\mathcal{T}\))}\\
\rule{\columnwidth}{0.5pt}\\
Input: \> \> \> A set $X$, $|X|\ge 3$, and a dense set $\T$ of 
trinets on $X$. \\
Output: \> \> \>A 1-nested network $N$ on $X$ with $Tr(N)=\T$, or \\
\> \> \> the statement ``There is no 1-nested network displaying $\T$".\\
\rule{\columnwidth}{0.5pt}\\
1. \> $Stack = \emptyset$, $G=X$ \\
2. \> While there is some cherry $S$ in $\T$ 
with $z\in S$ or some cactus \\
3. \> 
 $H=(a_1,\dots,a_p: b_1,\dots, b_q:z)$ with support 
$S=\{a_1,\dots,a_p, b_1,\dots, b_q,z\}$ \\
4.\>in $\T$
do\\
5. \>  \> Put the symbol $R_{z:S}$ on the top of $Stack$. \\
6. \> \>  If $|G-( S-\{z\})| \le 2$, then let $N$ be either the bush on 
$G$ or \\
7. \> \> the two-leafed network on $G$, depending on $\T$.\\
8. \>  \> Let $\T = \T - \T_{z:S}$, $G=G-(S-\{z\})$.  \\
9. \> end ``do (line 2)" \\
10. \> If  $|G| \ge 3$, then output  ``There is no 1-nested network 
displaying $\T$" \\
11. \>else do\\
12. \> \> While there is some $R_{z:S}$ on the top of $Stack$, do
$N = R^{-1}_{z:S}(N)$.  \\
13. \> \> Output $N$\\
13. \> end ``do (line 12)" \\
\end{tabbing}}
\caption{Pseudo-code for an algorithm to 
construct a 1-nested network from
a dense set of trinets, or decide that 
such a network does not exist.}
\label{algorithm:multicons}
\end{figure}

We now use the algorithms \textsc{FindCherry} 
and \textsc{FindCactus} to show that 
it can be decided in polynomial time whether
or not a dense set 
of trinets is displayed by 
a 1-nested network using the algorithm presented in 
Fig.~\ref{algorithm:multicons}.

\begin{theorem} \label{build}
For $X$ with $|X| \ge 3$ and $\T$ a dense
set of trinets on $X$, algorithm \textsc{BuildNet} has 
run-time that is polynomial in $|X|$ and is correct.
\end{theorem}

\begin{proof}
Algorithm \textsc{BuildNet} has run-time that is polynomial in $|X|$ 
since the check required in line 2 can be executed
in time that is polynomial in $|X|$ by 
Lemmas~\ref{polyfind1} and \ref{polyfind2}

Now, if  algorithm \textsc{BuildNet} outputs 
``There is no 1-nested network displaying $\T$", then 
by Proposition~\ref{cherryorcactus}, 
Lemma~\ref{fatcherry} and Lemma~\ref{cactus},
there is no 1-nested network $N$ on $X$ with $Tr(N)=\T$.
Moreover, if \textsc{BuildNet} outputs a network $N$, then 
$N$ is clearly 1-nested, and $Tr(N)=\T$ by (\ref{identity}). 
This completes the proof.
\end{proof}

\begin{remark}
Although we have shown that  algorithm \textsc{BuildNet} 
has run-time that is polynomial in $|X|$, it could be of 
interest to see if faster, more sophisticated algorithms
can be developed.
\end{remark}

\section{Discussion}\label{discussion}

In this paper, we have shown that we can recover a 1-nested network from 
`perfect data', viz. the dense set of 1-nested trinets
that is displayed by the network.
In practice, we will not 
usually have access to such information for  
biological datasets. Even so, it should be quite 
straight-forward to at least 
compute a dense set of trinets for 
any given biological dataset using existing phylogenetic
network methods.
For example, given a multiple sequence alignment, one could compute the
most parsimonious or most likely trinet for
every sub-alignment of 3 sequences (using, e.g. methods
described in \cite{JNST06,JNST09}), 
which would be feasible as there are a bounded number of 
1-nested trinets. Note that this would have the advantage that no
`breakpoints' would need to be computed for the
multiple alignment, which is 
a first (and sometimes quite difficult) step 
that is usually required when constructing phylogenetic 
networks from phylogenetic trees (cf. e.g. 
\cite[Chapter 11]{HRS11}, \cite[Section 2]{N11}).

Given that computing dense sets of trinets is feasible
for biological data, 
it could be reasonable to develop methods for finding
1-nested networks displaying as
many trinets as possible from a dense set of trinets.
Similar techniques have been developed 
for triplets e.g. \cite{HIKS11,IKKSHB09,TEPA11},
although it is worth noting that it is NP-hard to 
find a tree displaying a maximum number of rooted triplets from 
an arbitrary set of triplets \cite{B97,J01,W04}
(even if the set is dense \cite{BGJ10}).
Alternatively, it might be of interest to 
investigate if there might be  
an `Aho-type' algorithm  \cite{ASSU81} to determine if
an arbitrary subset of 1-nested trinets encodes a 1-nested network,
and, if so, adapt this to give `Min-Cut' type algorithms 
for building 1-nested networks from sets of trinets
(cf. \cite{P02,SS00,SS03}).
A first step in this direction could be to determine whether or not 
it is an NP-complete problem to decide if
an arbitrary subset of 1-nested trinets encodes a 1-nested network
(in particular, note that there are non-dense sets 
of 1-nested trinets that 
encode 1-nested networks -- e.g. the 1-nested network $N$ on $\{w,x,y,z\}$
pictured in  Fig.~\ref{encode-cactus}(a)
is the only 1-nested
network on $\{w,x,y,z\}$ displaying the two trinets presented 
in Fig.~\ref{encode-cactus}(b)). 
\begin{figure}[t] \centering
\includegraphics[scale=0.5]{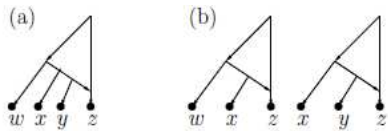}
\caption{The 1-nested network $N$ on
$\{w,x,y,z\}$ depicted in (a)
is uniquely determined by the two trinets  
pictured in (b). As before, directions 
are omitted for clarity when clear. Also only the vertices that
are leaves are marked by a dot.
 }
\label{encode-cactus}
\end{figure}

In another direction, clearly 
we can ask for results along the lines of 
those presented above for {\em level-$k$ networks} \cite{CJ05}, $k \ge 2$, 
phylogenetic networks that have a bounded level of complexity 
depending on $k$ (and also, of course,  `$k$-nested' networks). 
Note that there are non-recoverable level-2
networks (e.g. Fig.~\ref{not-recoverable}), and 
so this could be rather more technical.
Moreover, it should be noted that, for $k \ge 3$,
there are level-$k$ networks that are not
of level-$(k-1)$ all of whose trinets have 
fixed level (see Fig.~\ref{not-level-preserving}).
Thus, the levels of the trinets 
displayed by a network do not necessarily 
determine the level of a network. For
practical purposes, it might also be of interest to 
determine a way to enumerate the level-$k$ trinets, $k\ge 2$.
\begin{figure}[t] \centering
\includegraphics[scale=0.5]{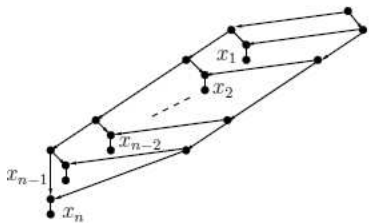}
\caption{A level-$n$
phylogenetic network $N$ on $\{x_1,x_2,\dots, x_n\}$, $n\geq 4$, 
for which every trinet in $Tr(N)$ is of level-$3$. For clarity arc directions
are omitted when clear.}
\label{not-level-preserving}
\end{figure}

Another avenue worth exploring, could be to try generalizing
the above results to `$r$-nets', $r \ge 4$, i.e.
phylogenetic networks with $r$-leaves
(note that in case $r=4$ quartet trees are commonly used 
to build phylogenetic trees, e.g. \cite{SH96}). Note that it is
straight-forward to extend Definition~\ref{tripdisp} to  
obtain a set of $r$-nets displayed by a phylogenetic network.
This could be quite useful in practice since
it might be possible to obtain more accurate estimates 
for $r$-nets than trinets (at least for 
$r=4$) before we try to piece them together, although,
technically speaking, this could be very challenging.

Finally, we conclude with what we consider to be a 
rather bold conjecture:

\begin{conjecture}
If $N$ is a recoverable phylogenetic network on $X$, then $Tr(N)$
encodes $N$, that is, if $N'$ a recoverable phylogenetic network on $X$ 
such that $Tr(N)=Tr(N')$ then $N$ is isomorphic to $N'$.
\end{conjecture}

A first (and probably quite instructive!) `exercise' 
could be to try and show that this conjecture at least 
holds for level-2 networks. Note
that if this conjecture were true, then 
as in Corollary~\ref{metric}, we 
would immediately obtain a new 
proper metric on the set of recoverable
phylogenetic networks on $X$.\\

\noindent {\bf Acknowledgements}
VM thanks the Royal Society for supporting his visit to New Zealand.
Both authors thank  Mike Hendy, David Penny, 
Charles Semple, Peter Stadler and
Mike Steel for hosting them during their sabbatical, during which 
this work was conceived and undertaken.

\end{document}